\documentclass[11pt]{article}
\usepackage{authblk} 
\usepackage{fullpage}
\usepackage{amsmath}
\usepackage{amsfonts}
\usepackage{amssymb}
\usepackage{xargs}

\newcommandx{\xxRightarrow}[1]
  {\xRightarrow
  {\protect{
  \raisebox{ -0.3pt }[0pt][0pt]{ \ensuremath{ \scriptscriptstyle{ #1 } }}}
  }
}

\usepackage{hyperref}
\hypersetup{hidelinks}
\usepackage[inline]{enumitem}
\setlist{nolistsep}

\usepackage{dsfont}
\usepackage{amsthm} 
\usepackage{cleveref}
\usepackage{scalerel}
\usepackage{mathtools} 
\usepackage{stackengine}
\usepackage{tikz} 
\usepackage{pgfplots} 
\usetikzlibrary{math} 
\pgfplotsset{compat=1.17}
\usetikzlibrary{arrows.meta} 
\usepackage{xcolor} 
\definecolor{blue1}{RGB}{70, 150, 190}
\definecolor{blue5}{RGB}{0, 0, 75}
\tikzset{%
  dashes/.style args={#1}{%
    dash pattern=on #1 off #1
  }
}

\newtheoremstyle{mystyle}
  {}
  {}
  {\itshape}
  {}
  {\bfseries}
  {.}
  { }
  {\thmname{#1}\thmnumber{ #2}\thmnote{ (#3)}}

\theoremstyle{mystyle}

\newtheorem{regtheorem}{Theorem}

\newtheorem{posttheorem}{Theorem}

\newtheorem*{mytheorem}{Theorem}
\newtheorem{mycor}{Corollary}
\newtheorem{mylemma}{Lemma}
\newtheorem{myclaim}{Claim}
\newtheorem{myfact}{Fact}

\newcommand{\bra}[1]{\left<#1\right|}
\newcommand{\ket}[1]{\left|#1\right>}

\DeclareMathOperator{\tr}{tr\!}
\DeclareMathOperator{\id}{\mathds{1}}
\newcommand{\h}{\mathcal{H}}
\newcommand{\K}{\mathcal{K}}

\crefname{regtheorem}{theorem}{theorems}    
\Crefname{regtheorem}{Theorem}{Theorems}
\crefname{posttheorem}{theorem}{theorems}    
\Crefname{posttheorem}{Theorem}{Theorems}
\title{Notes on distinguishability of postselected computations}
\author{Zuzana Gavorov\'a}
\affil{School of Computer Science and Engineering,
The Hebrew University of Jerusalem, Jerusalem, Israel}
\date{}

\begin{document}
\linespread{1.05}
\lineskiplimit=-5pt
\maketitle
\begin{abstract}

    The framework of postselection is becoming more and more important in various recent directions in Quantum Computation research. Postselection renders simple computational models able to perform general quantum computation. This was first observed for the linear optics model [E. Knill, R. Laflamme, G. J. Milburn, Nature 409, 46 (2001)], and has since provided us with many near-term candidates for the quantum advantage, commuting computations [M. J. Bremner, R. Jozsa, D. J. Shepherd, Proc. R. Soc. A 467, 459 (2011)] being the first. To facilitate the discussion of errors in the presence of postselection, we define and characterize trace-induced distance and diamond distance of postselected computations. We show counterexamples to simple properties that one would expect of any distance measure; the properties of convexity (when considering only the pure-state inputs would suffice), contractivity, and subadditivity of errors. On the positive side, we prove that certain weaker versions of contractivity and subadditivity and a number of other properties are preserved in the postselected setting. We achieve this via a "conversion lemma" that translates any inequality from the standard to the postselected setting.
\end{abstract}
\maketitle

\section{Introduction}

The postselected setting has recently drawn a lot of attention. Postselection renders linear optics universal for quantum computation \cite{knill2001scheme} and larger gate errors tolerable in fault-tolerant quantum computation \cite{knill2005quantum,reichardt2006error,aliferis2008accuracy}.
Arguments using postselection \cite{aaronson2005quantum} provide the evidence for quantum supremacy, i.e. that we cannot simulate classically certain quantum computations: commuting quantum computations \cite{bremner2011classical,bremner2017achieving}, boson sampling \cite{aaronson2011computational}, one clean qubit DQC1 computations \cite{morimae2014hardness}, random circuits \cite{bouland2019complexity}.  These computations are, moreover, quite simple and believed to be achievable in the near future - by noisy intermediate term quantum (NISQ) devices.

In this note we define postselection equivalents of trace-induced distance and diamond distance. We pick some known inequalities between the standard distance measures, also adding a new one, and prove that equivalent relations hold in the postselected setting. 

The difficulty is that in the postselected setting it only makes sense to compare computations after their output has been renormalized. Let $ L(\h)$ be the set of linear operators from the finite-dimensional\footnote{All Hilbert spaces throughout this note are finite-dimensional.} Hilbert space $\h$ to itself. We say that a map $\Phi: L(\h)\to L(\h')$ is a {\it postselection superoperator} if 
\begin{enumerate*}[label=(\arabic*)]
\item it is linear, completely positive (CP), trace-nonincreasing and 
\item its postselection probability $\tr\left[{\Phi}(\rho)\right]$ is nonzero for all $\rho\in\mathcal{D}(\h)$,
\end{enumerate*} where $\mathcal{D}(\h)\subset L(\h)$ is the set of normalized density operators. A postselection superoperator $\Phi$ followed by the renormalisation of the output corresponds to the map
\begin{equation}
    \rho\mapsto\frac{{\Phi}(\rho)}{\tr\left[{\Phi}(\rho)\right]}\text{,}\label{map_nonlinear}
\end{equation}

\noindent which introduces a non-linearity in $\rho$. The distances of such maps break certain relations: we will see counterexamples to convexity, contractivity and subadditivity. However, under certain conditions translation of some standard inequalities to the postselected setting is possible. We achieve this via our "conversion lemma", which states that under these conditions the postselection probability is almost constant in $\rho$, so that the nonlinear map \eqref{map_nonlinear} can be replaced by a linear one.

Errors in the postselected setting were previously investigated \cite{beverland2020lower,aaronson2019online,gao2015quantum,aaronson2006qma} for some specific applications requiring a version of subadditivity (also called union bound). We take the more general approach of deriving postselection equivalents of the following well-established results on diamond distance.

\subsection{Preceding work: 
Two inequalities to carry over to the postselected setting
}

Kitaev \cite{kitaev1997quantum} defined diamond distance $d_\diamondsuit$ and Aharonov, Kitaev and Nisan proved \cite[Theorem 4]{mixed_states_aharonov98} the following property of $d_\diamondsuit$, useful when evaluating how well a computation composed of subroutines simulates an ideal computation composed of ideal subroutines:

\begin{regtheorem}[due to \cite{mixed_states_aharonov98}: Subadditivity of $\boldsymbol{d_\diamondsuit}$]\label{lem:chaining_ah}
	Let the simulating computation use $N$ subroutines, each at most $\epsilon$-far in $d_\diamondsuit$ from its ideal. Then the simulating computation is at most $N\epsilon$-far in $d_\diamondsuit$ from the ideal computation.
\end{regtheorem}

The next result from Watrous \cite{watrous2018theory} relates $d_\diamondsuit$ to what we call the {\it operational} trace-induced distance $d_{tr}^\mathcal{D}$ of two linear trace-nonincreasing CP maps $\Phi, \Psi: L(\h)\to L(\h')$:
\begin{equation*}
    d_{tr}^\mathcal{D}(\Psi, \Phi):=\sup_{\rho\in\mathcal{D(\h)}}\left|\!\left|\Psi(\rho)-\Phi(\rho)\right|\!\right|_{tr}\text{,}
\end{equation*}
\noindent where $\left|\!\left|\cdot\right|\!\right|_{tr}$ is the trace norm. This is not the same as the usual trace-induced distance $d_{tr}$, but $d_\diamondsuit$ is the stabilized version of both $d_{tr}$ and $d_{tr}^\mathcal{D}$ (see \Cref{sec:prelim}).
Watrous \cite[Theorem 3.56]{watrous2018theory} proved that under a certain condition upper bounding the stabilized version is 'for free':

\begin{regtheorem}[due to \cite{watrous2018theory}: $\boldsymbol{d_{tr}}$ small $\boldsymbol{\xxRightarrow{U} d_\diamondsuit}$ small]\label{lem:diam_watrous}
    Let $U:\h\to \h'$ be an isometry, $\Psi: L(\h)\to L(\h')$ a linear CP trace-preserving map.
    If $d_{tr}^\mathcal{D}(\Psi, U\cdot U^\dagger)\leq\epsilon$, then $d_\diamondsuit(\Psi, U\cdot U^\dagger)\leq\sqrt{2\epsilon\,}$.
\end{regtheorem}

\noindent The condition is that one of the superoperators corresponds to an isometry. In \Cref{lem:diam_watrous} and throughout we indicate this condition by the superscript $U$ on the implication symbol.

\subsection{Overview of the results}

\subsubsection*{A new inequality for the standard setting}

We add an inequality that uses Stinespring dilation to relate $d_{tr}^\mathcal{D}$ to the operator norm $\left|\!\left|\cdot\right|\!\right|_{op}$. (The fact that $d_{tr}^\mathcal{D}(A\cdot A^\dagger,B\cdot B^\dagger) \leq 2\left|\!\left|A-B\right|\!\right|_{op}$ if $\left|\!\left|A\right|\!\right|_{op}, \left|\!\left|B\right|\!\right|_{op}\leq 1$ follows from Lemma 12.6 of \cite{mixed_states_aharonov98}, here we prove an opposite bound.) Recall that by Stinespring dilation any linear CP map $\Psi: L(\h)\to L(\h')$ can be written as $\Psi(\cdot)=\tr_{\,\K'}\left[A\cdot A^\dagger\right]$ for some (non-unique) linear operator $A:\h\to\h'\otimes\K'$.

\newtheorem*{lemun}{\Cref{lem:op}}
\begin{lemun}[$\boldsymbol{d_{tr}^\mathcal{D}}$ small $\boldsymbol{\xxRightarrow{U}}$ $\boldsymbol{\left|\!\left|\cdot\right|\!\right|_{op}}$ small) (roughly] 
Let $U\!:\h\to \h'$ be an isometry and $\Psi: L(\h)\to L(\h')$ a linear CP trace-nonincreasing map. Denote by $A:\h\to\h'\otimes\K'$ some Stinespring-dilation operator of $\Psi$. If $d_{tr}^\mathcal{D}(\Psi, U\cdot U^\dagger)\leq\epsilon$, then there exists $\ket{g}\in\K'$ such that $\left|\!\left|A-U\otimes \ket{g}\right|\!\right|_{op}\leq 2\sqrt{\epsilon}$.
\end{lemun}

\noindent To prove \Cref{lem:op} we first prove that $d_{tr}(\Psi, \Phi) \leq 2 d_{tr}^\mathcal{D}(\Psi, \Phi)$ for any pair of linear maps $\Psi, \Phi: L(\h)\to L(\h')$ (\Cref{lem:trace_dists}). A corollary of \Cref{lem:op} is a new version of \Cref{lem:diam_watrous} with the trace-preserving assumption on $\Psi$ relaxed (\Cref{lem:diam}).

\subsubsection*{The postselected setting}

We define postselection equivalents of $d_{tr}^\mathcal{D}$ and $d_\diamondsuit$ in the following way: Let ${\Psi, \Phi: L(\h)\to L(\h')}$ be postselection superoperators, then
    \begin{eqnarray*}
        \widehat{d}_{tr}(\Psi, \Phi) &:=& \sup_{\rho\in \mathcal{D}(\h)}\left|\!\left|\frac{\Psi(\rho)}{\tr\left[\Psi(\rho)\right]}-\frac{\Phi(\rho)}{\tr\left[\Phi(\rho)\right]}\right|\!\right|_{tr}\\
        \widehat{d}_{\diamondsuit}(\Psi, \Phi) &:=& \sup_{\K}\, \widehat{d}_{tr}(\Psi\otimes\mathcal{I}_\K, \Phi\otimes\mathcal{I}_\K)\text{,}
    \end{eqnarray*}
\noindent where $\mathcal{I}_\K: L(\K)\to L(\K)$ is the identity superoperator. Note that $\widehat{d}_{tr}$ and $\widehat{d}_{\diamondsuit}$ are pseudometrics: they are zero if $\Psi=\Phi$ (but not only if!), they are symmetric and obey the triangle inequality. Unfortunately, the function of $\rho$ that $\widehat{d}_{tr}$ maximizes could be not convex (see \Cref{lem:convex}). We will still show that both suprema are achieved and that the pseudometrics obey inequalities equivalent to \Cref{lem:chaining_ah,lem:diam_watrous,lem:op}:

\begin{mytheorem}[roughly] 
Let $U:\h\to \h'$ be an isometry, $\Psi: L(\h)\to L(\h')$ a post\-selection superoperator. Denote by $A:\h\to\h'\otimes\K'$ a Stinespring-dilation operator~of~$\Psi$.
\begin{enumerate}[topsep=2pt,itemsep=2pt]
    \item[\textbf{\upshape{\ref{lem:chaining_post}}}] \textbf{\upshape{(Weak subadditivity of $\boldsymbol{\widehat{d}_\diamondsuit}$).}} If the ideal subroutines are trace-preserving, \Cref{lem:chaining_ah} holds also when the distances are measured by $\smash{\widehat{d}_\diamondsuit}$.
    \item[\textbf{\upshape{\ref{lem:diam_post}}}] \textbf{\upshape{($\boldsymbol{\widehat{d}_{tr}}$ small $\boldsymbol{\xxRightarrow{U} \widehat{d}_\diamondsuit}$ small).}}
    \Cref{lem:diam_watrous} has a postselection equivalent.
    \item[\textbf{\upshape{\ref{lem:op_post}}}] \textbf{\upshape{($\boldsymbol{\widehat{d}_{tr}}$ small $\boldsymbol{\xxRightarrow{U} \left|\!\left|\cdot\right|\!\right|_{op}}$ small).}}
    \Cref{lem:op} has a postselection equivalent.
\end{enumerate}
\end{mytheorem}

\noindent Remarkably, the trace-preserving requirement in \Cref{lem:chaining_post} is not superfluous; once dropped, a counterexample exists (see \Cref{sec:additivity}). \Cref{lem:diam_post,lem:op_post} are proven using our main result - the Conversion \namecref{lem:conversion}:

\newtheorem*{lemmas}{\Cref{lem:conversion}}
\begin{lemmas}[Main: Conversion lemma) (roughly]
	Let $\Psi,\Phi$ be postselection superoperators, \linebreak $\Phi$ trace-preserving. There exists a positive scalar $k$ such that $\widehat{d}_{tr}(\Psi, \Phi)$ and $d_{tr}^\mathcal{D}(\frac{\Psi}{k}, \Phi)$ are equivalent up to constant factors.
\end{lemmas}

\noindent As a (zero-distance) example take $\Phi: L(\h)\to L(\h')$ to be the identity superoperator between two spatially separated registers $\h$ and $\h'$ of the same dimension. The quantum teleportation with postselection on the appropriate Bell-measurement outcome is a local $\Psi$ that simulates the nonlocal $\Phi$ {\it exactly}. By \Cref{lem:conversion} we have $\frac{\Psi}{k}=\Phi$, so $\Psi$'s postselection probability must be $k$, i.e constant over all input states - indeed we know that $\smash{\tr\left[\Psi_{\text{telep}}(\rho)\right]=k_{\text{telep}}=\dim(\h)^{-2}}$.

\section{Preliminaries}\label{sec:prelim}

In this section we review some standard distance measures on linear trace-nonincreasing CP superoperators. Let $\Psi, \Phi: L(\h)\to L(\h')$ be such superoperators. The {\it operational trace induced distance} is induced from a norm $\left|\!\left|\cdot\right|\!\right|_{tr}^\mathcal{D}$, which, in turn, is induced from the trace norm on operators in the following way
\begin{equation}
    d_{tr}^\mathcal{D}(\Psi, \Phi)=\left|\!\left|\Psi-\Phi\right|\!\right|_{tr}^\mathcal{D}:=\sup_{\rho\in\mathcal{D}(\h)}\left|\!\left|\Psi(\rho)-\Phi(\rho)\right|\!\right|_{tr}\text{.}\label{eq:def_op_tr}
\end{equation}

\noindent The {\it trace-induced distance} is similar, the only difference being that the supremum is now over all linear operators $ L(\h)$, not only density matrices
\begin{equation}
    d_{tr}(\Psi, \Phi)=\left|\!\left|\Psi-\Phi\right|\!\right|_{tr}:=\!\!\!\sup_{\substack{X\in L(\h),\\\left|\!\left|X\right|\!\right|_{tr}\leq 1}}\!\!\!
    \left|\!\left|\Psi(X)-\Phi(X)\right|\!\right|_{tr}\text{.}\label{eq:def_tr}
\end{equation}
By definition, this distance upper bounds the first one, $d_{tr}^\mathcal{D}(\Psi, \Phi)\leq d_{tr}(\Psi, \Phi)$. On one hand, Watrous \cite{watrous2005notes} found $\Psi, \Phi$ such that this inequality is strict, but on the other hand, the {\it diamond distance}, $d_\diamondsuit(\Psi, \Phi)=\left|\!\left|\Psi-\Phi\right|\!\right|_\diamondsuit$, is the stabilized version of both
\begin{eqnarray}
    d_\diamondsuit(\Psi, \Phi):&=&\sup_{\K}d_{tr}(\Psi\otimes\mathcal{I}_\K, \Phi\otimes\mathcal{I}_\K)\nonumber\\
    &=&\sup_{\K}d_{tr}^\mathcal{D}(\Psi\otimes\mathcal{I}_\K, \Phi\otimes\mathcal{I}_\K)\text{,}\label{eq:def_diam}
\end{eqnarray}

\noindent where $\mathcal{I}_\K: L(\K)\to L(\K)$ is the identity superoperator. The last equality in (\ref{eq:def_diam}) follows from the following result of Watrous \cite[Lemma 3.45 and Theorem 3.51]{watrous2018theory} by setting $f=\Psi-\Phi$:
\begin{myfact}[due to \cite{watrous2018theory}]\label{lem:diam_max}
    For any Hermitian-preserving linear map $f: L(\h)\to L(\h')$, there exists a normalized vector $\ket{u}\in\h\otimes\h$ such that setting $\K=\h$ and $X=\ket{u}\!\!\bra{u}\in\mathcal{D}(\h\otimes\h)$ achieves the suprema in definitions \eqref{eq:def_diam} and \eqref{eq:def_tr}, i.e. $\left|\!\left|f\right|\!\right|_\diamondsuit=\left|\!\left|f\otimes\mathcal{I}_\h(\ket{u}\!\!\bra{u})\right|\!\right|_{tr}$.
\end{myfact}

We have mentioned the {\it Stinespring dilation}\cite{stinespring1955positive}, which states that any linear CP superoperator $\Psi: L(\h)\to L(\h')$ can be written (non-uniquely) as $\Psi(\cdot)=\tr_{\,\K'}\left[A\cdot A^\dagger\right]$ for some extension Hilbert space $\K'$ and some linear operator $A:\h\to\h'\otimes\K'$. For linear operators $O:\h\to\h'$ we will use the {\it operator norm} defined as $\left|\!\left|O\right|\!\right|_{op}:=\sup_{\ket{v}\in\h,\, \left|\!\left|\ket{v}\right|\!\right|=1}\left|\!\left|O\ket{v}\right|\!\right|$. We will find the following \namecref{lem:norms} useful:
\begin{myfact}\label{lem:norms}
    If $\Psi$ is a linear completely positive superoperator and $A$ corresponds to its Stinespring dilation, then $\left|\!\left|\Psi\right|\!\right|_\diamondsuit = \left|\!\left|\Psi\right|\!\right|_{tr}=\left|\!\left|\Psi \right|\!\right|_{tr}^\mathcal{D}=\left|\!\left|A\right|\!\right|^2_{op}$.
\end{myfact}

\begin{proof} Note that $\left|\!\left|\Psi\right|\!\right|_\diamondsuit \geq \left|\!\left|\Psi\right|\!\right|_{tr}\geq\left|\!\left|\Psi \right|\!\right|_{tr}^\mathcal{D}\geq\left|\!\left|A\right|\!\right|^2_{op}$, the first two inequalities following immediately from the definitions and the last from the fact that for a positive semidefinite operator $X$ we have $\left|\!\left|X\right|\!\right|_{tr}=\tr\left[X\right]$ and so $\left|\!\left|\Psi \right|\!\right|_{tr}^\mathcal{D}=\sup_{\rho\in\mathcal{D}(\h)}\tr\left[A\rho A^\dagger\right]$,
which is bigger or equal to ${\left|\!\left|A\right|\!\right|_{op}^2= \sup_{\ket{v}\in\h,\, \left|\!\left|\ket{v}\right|\!\right|=1}\tr\left[A\ket{v}\!\!\bra{v}A^\dagger\right]}$.
It remains to show $\left|\!\left|\Psi\right|\!\right|_\diamondsuit\leq \left|\!\left|A\right|\!\right|_{op}^2$. Apply \Cref{lem:diam_max} to get $\left|\!\left|\Psi\right|\!\right|_\diamondsuit = \tr\left[\Psi\otimes\mathcal{I}_\h\left(\ket{u}\!\!\bra{u}\right)\right]\leq \left|\!\left|A\otimes\id\right|\!\right|_{op}^2=\left|\!\left|A\right|\!\right|_{op}^2$, which completes the proof.
\end{proof}

\noindent If $\Psi$ is also trace-nonincreasing, we have $\left|\!\left|\Psi(\rho)\right|\!\right|_{tr}=\tr\left[\Psi(\rho)\right]\leq 1$ for all $\rho\in\mathcal{D}(\h)$ so that by definition \eqref{eq:def_op_tr} $\left|\!\left|\Psi\right|\!\right|_{tr}^\mathcal{D}\leq 1$. By definition \eqref{eq:def_tr} $\left|\!\left|\Psi(X)\right|\!\right|_{tr}\leq \left|\!\left|\Psi\right|\!\right|_{tr} \left|\!\left|X\right|\!\right|_{tr}$, which with \Cref{lem:norms} gives:

\begin{myfact} A trace non-increasing CP map $\Psi$ contracts the trace norm: ${\left|\!\left|\Psi(\cdot)\right|\!\right|_{tr}\leq \left|\!\left|\cdot\right|\!\right|_{tr}}$.\label{lem:contract}
\end{myfact}

\section{Results: Standard setting}

In this section we prove \Cref{lem:op} and introduce a new version of \Cref{lem:diam_watrous}.

\begin{regtheorem}[$\boldsymbol{d_{tr}^\mathcal{D}}$ small $\boldsymbol{\xxRightarrow{U} \left|\!\left|\cdot\right|\!\right|_{op}}$ small]\label{lem:op}
    Let ${U\!\!:\h\to\! \h'}$ be an isometry and $\Psi: L(\h)\to L(\h')$ a linear CP trace-nonincreasing superoperator. Denote by $A:\h\to\h'\otimes\K'$ a Stinespring-dilation operator of $\Psi$. If $d_{tr}^\mathcal{D}(\Psi, U\cdot U^\dagger) \leq \epsilon$, then there exists $\ket{g}\in\K'$, with the norm $1-\epsilon\leq \left|\!\left|\ket{g}\right|\!\right|^2\leq \left|\!\left|A\right|\!\right|^2_{op}\leq 1$, such that $\left|\!\left|A-U\otimes \ket{g}\right|\!\right|_{op}\leq 2\sqrt {\epsilon}$.
\end{regtheorem}

\noindent As a corollary we get a version of \Cref{lem:diam_watrous} with a worse bound but working also for $\Psi$s that decrease the trace of some inputs. This we need for the next section.
\begin{mycor}\label{lem:diam}
    Let $U:\h\to \h'$ be an isometry and $\Psi: L(\h)\to L(\h')$ a linear CP trace-nonincreasing superoperator.
    If $d_{tr}^\mathcal{D}(\Psi, U\cdot U^\dagger) \leq \epsilon$, then $d_\diamondsuit(\Psi, U\cdot U^\dagger)\leq 4\sqrt{\epsilon}+\epsilon$.
\end{mycor}

\noindent The \namecref{lem:diam} follows from $d_\diamondsuit(\Psi, U\cdot U^\dagger)\leq d_\diamondsuit(\Psi, \left|\!\left|\ket{g}\right|\!\right|^2 U\cdot U^\dagger)+ |\,\left|\!\left|\ket{g}\right|\!\right|^2-1|$ (an application of the triangle inequality for $d_\diamondsuit$), from the fact that partial trace contracts $d_\diamondsuit$ and from Lemma 12.6 of \cite{mixed_states_aharonov98}, according to which $d_{\diamondsuit}(A\cdot A^\dagger,B\cdot B^\dagger) \leq 2\left|\!\left|A-B\right|\!\right|_{op}$ if $\left|\!\left|A\right|\!\right|_{op}, \left|\!\left|B\right|\!\right|_{op}\leq 1$.

To prove \Cref{lem:op} we first prove the following:

\begin{mylemma}\label{lem:trace_dists} 
Let $\Psi, \Phi: L(\h)\to L(\h')$ be two linear maps.
$d_{tr}(\Psi, \Phi)\leq 2\,d_{tr}^\mathcal{D}(\Psi, \Phi)$.
\end{mylemma}

\begin{proof}
For any pair $\ket{u},\ket{v}\in\h$ of unit vectors define the unnormalized $\ket{w_k}:=\ket{u}+i^k\ket{v}$ and observe that
\[
   \sum_{k=0}^3  \left|\!\left|\ket{w_k}\right|\!\right|^2=8 \quad \text{and} \quad \ket{u}\!\!\bra{v}=\frac{1}{4}\sum_{k=0}^3 i^k\ket{w_k}\!\!\bra{w_k}\text{,}
\]
because the roots of unity sum to zero, $\sum_{k=0}^3 i^k = 0$. For any linear $f: L(\h)\to L(\h')$ we get by the triangle inequality and the absolute homogeneity of trace norm
\begin{eqnarray*}
    \left|\!\left|f(\ket{u}\!\!\bra{v})\right|\!\right|_{tr}
    &\leq& \frac{1}{4}\sum_{k=0}^3  \left|\!\left|f(\ket{w_k}\!\!\bra{w_k})\right|\!\right|_{tr}\\
    &\leq& \frac{1}{4}\left(\sum_{k=0}^3  \left|\!\left|\ket{w_k}\right|\!\right|^2\right)
    \left|\!\left|f\right|\!\right|_{tr}^\mathcal{D}
     = 2 \left|\!\left|f\right|\!\right|_{tr}^\mathcal{D}\text{.}
\end{eqnarray*}

\noindent By a convexity argument\footnote{Operators of the form $\ket{u}\!\!\bra{v}$ are the extreme points of the convex set ${\{X\in L(\h); \left|\!\left|X\right|\!\right|_{tr}\leq 1\}}$ and $\left|\!\left|f(\cdot)\right|\!\right|_{tr}$ is a convex function; it maps convex set to a convex set and extreme points to extreme points.} there exist $\ket{u^*},\ket{v^*}\in\h$ that achieve the supremum of Eq. (\ref{eq:def_tr}). We get $\left|\!\left|f\right|\!\right|_{tr}=\left|\!\left|f(\ket{u^*}\!\!\bra{v^*})\right|\!\right|_{tr}\leq 2\left|\!\left|f\right|\!\right|_{tr}^\mathcal{D}$. Setting $f=\Psi-\Phi$ completes the proof. 
\end{proof}

\begin{proof}[Proof of \Cref{lem:op}]
By the definition of $d_{tr}$ in (\ref{eq:def_tr}) we have for all normalized $\ket{u},\ket{v}\in\h$
\begin{eqnarray*}
	d_{tr}(\Psi, U\cdot U^\dagger) &\geq& \left|\!\left|\Psi\left(\ket{u}\!\!\bra{v}\right)-U\ket{u}\!\!\bra{v}U^\dagger\right|\!\right|_{tr}\\*
	&=& \left|\!\left|U^\dagger\Psi\left(\ket{u}\!\!\bra{v}\right)U-\ket{u}\!\!\bra{v}\right|\!\right|_{tr}\left|\!\left|\,\ket{v}\!\!\bra{u}\,\right|\!\right|_{op}\\*
	&\geq& \left|\!\left|U^\dagger\Psi\left(\ket{u}\!\!\bra{v}\right)U\ket{v}\!\!\bra{u}-\ket{u}\!\!\bra{u}\right|\!\right|_{tr}\\*
	&\geq& \left|\bra{u}U^\dagger\Psi\left(\ket{u}\!\!\bra{v}\right)U\ket{v}-1\right|
\end{eqnarray*}

\noindent where we inserted $\left|\!\left|\,\ket{v}\!\!\bra{u}\,\right|\!\right|_{op}=1$ into the second line. The last two inequalities are $\left|\!\left|X\right|\!\right|_{tr}\left|\!\left|O\right|\!\right|_{op}\geq \left|\!\left|XO\right|\!\right|_{tr}$ and $\left|\!\left|X\right|\!\right|_{tr}\geq |\tr\left[X\right]|$ that hold for any $X, O\in L(\h)$ (see \cite[Lemma 10]{mixed_states_aharonov98}). When $\ket{u}=\ket{v}$ we can write $d_{tr}^\mathcal{D}$ instead of $d_{tr}$ (see definition \eqref{eq:def_op_tr}). Writing out the Stinespring-dilation for $\Psi$ we get
\begin{eqnarray}
     d_{tr}(\Psi, U\cdot U^\dagger)&\geq& \left|\tr\left(\ket{g_u}\!\!\bra{g_v}\right)-1\right|\label{eq:garbage_prod}\\
     \text{with} \ket{g_u}&:=&\big(\bra{u}U^\dagger\otimes\id_{\K'}\big)A\ket{u}\label{eq:def_g}
\end{eqnarray}
From \Cref{lem:norms} and $\Psi$ being CP trace-nonincreasing we have
\begin{equation}
    \left|\!\left|A\ket{u}\right|\!\right|^2\leq \left|\!\left|A\right|\!\right|^2_{op}=\left|\!\left|\Psi\right|\!\right|_\diamondsuit\leq 1 \label{eq:A}
\end{equation}
From the $d_{tr}^\mathcal{D}$ equivalent of (\ref{eq:garbage_prod}) and from (\ref{eq:def_g}) and (\ref{eq:A}) we get
\[
    1-d_{tr}^\mathcal{D}(\Psi, U\cdot U^\dagger)\leq \left|\!\left|\ket{g_u}\right|\!\right|^2 \leq \left|\!\left|A\right|\!\right|^2_{op}\leq 1\text{.}
\]
So the norm of $\ket{g_u}$ is withing the bounds required by \Cref{lem:op}. For any normalized $\ket{u},\ket{v}\in\h$ we can use (\ref{eq:A}) again to upper bound
\begin{eqnarray*}
	\left|\!\left|A\ket{v}-U\ket{v}\otimes \ket{g_u}\right|\!\right|^2 &=&\left|\!\left|A\ket{v}\right|\!\right|^2 + \left|\!\left|\ket{g_u}\right|\!\right|^2 -2\Re\left[\big(\bra{v}U^\dagger\otimes\bra{g_u}\big)A\ket{v}\right]\nonumber\\
	&\leq& 2\left(1-\Re\left<g_u|g_v\right>\right)\\
	&\leq& 2\,d_{tr}(\Psi, U\cdot U^\dagger)
\end{eqnarray*}

\noindent where $\Re$ stands for the real part of the complex number and the last inequality follows from (\ref{eq:garbage_prod}). Combining with \Cref{lem:trace_dists} we find that $\left|\!\left|A-U\otimes \ket{g_u}\right|\!\right|_{op}\leq 2\sqrt{d_{tr}^\mathcal{D}(\Psi, U\cdot U^\dagger)}$ with $\ket{g_u}$ is defined by Eq. (\ref{eq:def_g}) for any normalized $\ket{u}\in\h$.
\end{proof}

\section{Results: Postselected setting}
In the introduction we defined the trace-induced distance and the diamond distance for the postselected setting. For any postselection superoperators $\Psi, \Phi: L(\h)\to L(\h')$ we had
\begin{eqnarray}
        \widehat{d}_{tr}(\Psi, \Phi)&:=&\sup_{\rho\in \mathcal{D}(\h)}\underbrace{\left|\!\left|\frac{\Psi(\rho)}{\tr\left[\Psi(\rho)\right]}-\frac{\Phi(\rho)}{\tr\left[\Phi(\rho)\right]}\right|\!\right|_{tr}}_{f_{\Psi, \Phi}(\rho)}\label{eq:def_tr_post}\\
        \widehat{d}_{\diamondsuit}(\Psi, \Phi)&:=&\sup_{\K}\, \widehat{d}_{tr}(\Psi\otimes\mathcal{I}_\K, \Phi\otimes\mathcal{I}_\K)\text{,}\label{eq:def_diam_post}
\end{eqnarray}
where we added the shorthand $f_{\Psi,\Phi}$ to denote the objective function of the maximisation in (\ref{eq:def_tr_post}). In general, this function is not convex.

\begin{figure}
    \centering
    \begin{tikzpicture}
    \begin{axis}[
    domain     = 0:1,
    samples=100,
    height=3cm,
    width=6.5cm,
    scale only axis,
    axis lines=center,  
    xtick={0.5,1},
    ticklabel style={align=center},
    xticklabels={{$1/2$\\${\scriptstyle 1/2\ket{0}\!\bra{0}+1/2\ket{1}\!\bra{1}}$},{1\\${\scriptstyle\ket{1}\!\bra{1}}$}},
    extra x tick style={align=center},
    extra x ticks={{0}},
    extra x tick labels={{0\\${\scriptstyle\ket{0}\!\bra{0}}$}},
    xlabel={{\\[-0.25em]$p$\\${\scriptstyle \rho\,=\,(1-p)\ket{0}\!\bra{0}}$\\[-0.1em]${\scriptstyle +p\ket{1}\!\bra{1}}$}},
    ylabel={$f_{{\Psi_\epsilon},{\Phi_\epsilon}}(\rho)$},
    xlabel style={below, align=center},
    ylabel style={right},
    xmin=-0.1,
    xmax=1.3,
    ymin=-0.05,
    ymax=2.2,
    legend style={draw=none,at={(1,0.35)},anchor=south},
    legend cell align={left},
      ]
    \pgfplotsinvokeforeach{32, 8, 4}{
    \addplot [mark = none,dashes={#1/6pt}]
     {(2*(1-x)*x*(1-2*(1/#1)))/((1-x)*x+(1/#1)*(1-2*x)^2-(1/#1)^2*(1-2*x)^2)};
    \addlegendentryexpanded{$\epsilon=1/#1$}
}
  \end{axis}
\end{tikzpicture}
    \caption{The objective function $f_{\Psi, \Phi}$ \eqref{eq:def_tr_post} is not always maximized at the pure states. For example when $\Psi=\Psi_\epsilon$ with probability $1-\epsilon$ projects its input onto $\ket{0}\!\!\bra{0}$ and $\Phi=\Phi_\epsilon$ onto $\ket{1}\!\!\bra{1}$ (see the proof of \Cref{lem:convex}).}
    \label{fig_concave}
\end{figure}
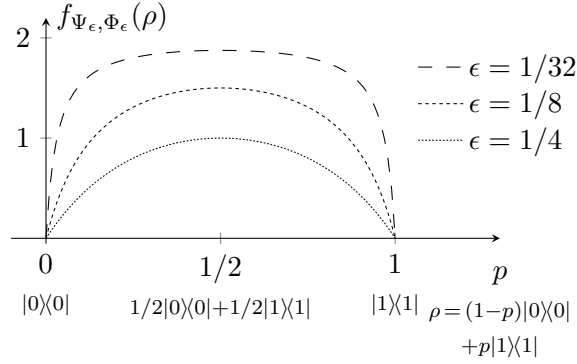

\begin{myclaim}\label{lem:convex}
    There exist $\Psi, \Phi$ such that $f_{\Psi,\Phi}$ is not convex on $\mathcal{D}(\h)$.
\end{myclaim}

\begin{proof}
Consider the maps ${\Psi_\epsilon}, {\Phi_\epsilon}:L(\mathbb{C}^2)\to L(\mathbb{C}^2)$ where ${\Psi_\epsilon}(\rho)=(1-\epsilon)\ket{0}\!\!\bra{0}\rho\ket{0}\!\!\bra{0}+\epsilon\ket{1}\!\!\bra{1}\rho\ket{1}\!\!\bra{1}$ and ${\Phi_\epsilon}(\rho)=(1-\epsilon)\ket{1}\!\!\bra{1}\rho\ket{1}\!\!\bra{1}+\epsilon\ket{0}\!\!\bra{0}\rho\ket{0}\!\!\bra{0}$, for some $\epsilon\in(0,\frac{1}{2})$. 
Both are postselection superoperators so $f_{{\Psi_\epsilon},{\Phi_\epsilon}}$ is well-defined on  $\mathcal{D}(\mathbb{C}^2)$. Observe that ${f_{{\Psi_\epsilon},{\Phi_\epsilon}}(\ket{0}\!\!\bra{0})=0}$, $f_{{\Psi_\epsilon},{\Phi_\epsilon}}(\ket{1}\!\!\bra{1})=0$, and $f_{{\Psi_\epsilon},{\Phi_\epsilon}}(\frac{1}{2}\ket{0}\!\!\bra{0}+\frac{1}{2}\ket{1}\!\!\bra{1})=2-4\epsilon$, therefore this $f_{{\Psi_\epsilon},{\Phi_\epsilon}}$ is not convex. See Fig. \ref{fig_concave}.
\end{proof}

It is still possible to replace the suprema in the definitions by maxima. In (\ref{eq:def_tr_post}) this is because $\mathcal{D}(\h)$ is compact and $f_{\Psi,\Phi}$ is continuous. The supremum in (\ref{eq:def_diam_post}) is attained by $\K=\h$, similarly as for its non-postselected predecessor, $\sup_{\K}d_{tr}^\mathcal{D}(\Psi\otimes\mathcal{I}_\K, \Phi\otimes\mathcal{I}_\K)$ \cite[Theorem 5]{watrous2005notes}. The proof is also similar, except for an additional step circumventing the possible non-convexity of $f_{\Psi,\Phi}$.

\begin{myclaim}\label{lem:diam_post_max}
    Let $\Psi, \Phi$ be postselection superoperators on $L(\h)$. Then $\widehat{d}_{\diamondsuit}(\Psi, \Phi)=\widehat{d}_{tr}(\Psi\otimes\mathcal{I}_\h, \Phi\otimes\mathcal{I}_\h)$.
\end{myclaim}

\begin{proof} The $\geq$ direction follows from definition \eqref{eq:def_diam_post}. We prove here the $\leq$ direction.

First, we claim that for any normalized $\ket{v}\in\h\otimes\K$, where $\K$ is of a higher dimension than $\h$, there exists a normalized $\ket{u}\in\h\otimes\h$ such that
\begin{equation}
    f_{\Psi\otimes\mathcal{I}_\K,\Phi\otimes\mathcal{I}_\K}(\ket{v}\!\!\bra{v})
    = f_{\Psi\otimes\mathcal{I}_\h,\Phi\otimes\mathcal{I}_\h}(\ket{u}\!\!\bra{u})\text{.}\label{eq:schmidt_trick}
\end{equation}
The proof is standard:
Since the Schmidt decomposition of $\ket{v}$ has at most $\dim(\h)$ terms, there exists $\ket{u}\in\h\otimes\h$ such that $\ket{v}=\left(\id_\h\otimes\, U\right) \ket{u}$ with $U:\h\to\K$ an isometry; $U^\dagger U=\id_\h$. Plugging this in for $\ket{v}$, noting that superoperators acting on different registers commute and that $\left|\!\left|(\id_\h\otimes U)\,\cdot\, (\id_\h\otimes U^\dagger)\right|\!\right|_{tr}=\left|\!\left|\cdot \right|\!\right|_{tr}$ (for example, by \Cref{lem:contract} applied in both directions) we get Eq. (\ref{eq:schmidt_trick})\footnote{The entire proof would stop here if the objective function was convex: Since pure states are the extreme points of the convex $\mathcal{D}(\h\!\otimes\!\K)$, we would have some $\ket{v}\!\!\bra{v}\in \mathcal{D}(\h\!\otimes\!\K)$ achieve each \protect{$\widehat{d}_{tr}$} in (\protect{\ref{eq:def_diam_post}}), similarly as in the footnote on the previous page.}.

Next, we prove that for any postselection superoperators $A, B:L(\mathcal{N})\to L(\mathcal{M})$ and for any $\rho\in \mathcal{D}(\mathcal{N})$ there exists $\ket{v}\in\mathcal{N}\otimes\mathcal{N}$ such that
\begin{equation}
    f_{A, B}(\rho)\leq 
    f_{A\otimes \mathcal{I}_{\mathcal{N}},B\otimes\mathcal{I}_{\mathcal{N}}}
    (\ket{v}\!\!\bra{v})\text{.}\label{eq:purification_trick}
\end{equation}
We choose a $\ket{v}$ that is a purification of $\rho$, substituting $\rho=\tr_\mathcal{N}(\ket{v}\!\!\bra{v})$ into the left-hand side. The inequality follows by commutation and by
$\left|\!\left|\tr_\mathcal{N}(\,\cdot\,) \right|\!\right|_{tr}\leq \left|\!\left|\cdot \right|\!\right|_{tr}$ (\Cref{lem:contract}).

Finally, we set $A=\Psi \otimes \mathcal{I}_{\K'}$, $B=\Phi \otimes \mathcal{I}_{\K'}$ in (\ref{eq:purification_trick}) and to the right-hand side we apply Eq. (\ref{eq:schmidt_trick}) with $\K=\K'\otimes\h\otimes\K'$. We get $f_{\Psi\otimes\mathcal{I}_{\K'},\Phi\otimes\mathcal{I}_{\K'}}(\rho)\leq \max_{\ket{u}\in\h\otimes\h} f_{\Psi\otimes\mathcal{I}_\h,\Phi\otimes\mathcal{I}_\h}(\ket{u}\!\!\bra{u})$
for any $\K'$ and $\rho\in\mathcal{\h\otimes\K'}$. Taking suprema over them completes the proof.
\end{proof}

\subsection{Weak subadditivity and weak contractivity}\label{sec:additivity}

In this section we discuss how $\widehat{d}_\diamondsuit$ behaves with respect to composition of superoperators.

\begin{posttheorem}[Weak subadditivity]\label{lem:chaining_post}
Let $\Psi, \Phi: L(\h)\to L(\h')$ and $\Psi',\Phi' : L(\h'\otimes\K')\to L(\h'')$ be postselection superoperators, $\Phi'$ being trace-preserving.
Then 
\begin{equation}
    \widehat{d}_\diamondsuit (\Psi'\!\circ \left( \Psi\otimes \mathcal{I}_{\K'}\right), \Phi'\!\circ \left(\Phi\otimes \mathcal{I}_{\K'}\right))
    \leq  \widehat{d}_\diamondsuit (\Psi', \Phi') + \widehat{d}_\diamondsuit (\Psi, \Phi)\text{,}\label{eq:chaining}
\end{equation}
\noindent where $\mathcal{I}_{\K'}: L(\K')\to L(\K')$ is the identity superoperator.
\end{posttheorem}

\noindent Setting $\K'$ trivial and $\Psi'=\Phi'=:\tau$ we immediately get that $\widehat{d}_\diamondsuit$ contracts under the composition with a {\it trace-preserving} postselection superoperator:

\begin{mycor}[Weak contractivity] \label{lem:contract_post}
\!With $\Psi, \Phi$ as before, let ${\tau: L(\h')\to L(\h'')}$ be a trace-pre\-serving postselection superoperator. Then $\widehat{d}_\diamondsuit (\tau \circ \Psi, \tau \circ \Phi) \leq  \widehat{d}_\diamondsuit (\Psi, \Phi)$.
\end{mycor}

\noindent Interestingly, \Cref{lem:contract_post}, and consequently \Cref{lem:chaining_post}, do not hold if $\tau$, $\Phi'$ are not trace-preserving. 
Consider the following two postselection superoperators, constant on $\rho\in\mathcal{D}(\mathbb{C}^3)$:
$\Psi(\rho)=\tfrac{1}{2}\ket{0}\!\!\bra{0}+\tfrac{1}{2}\ket{1}\!\!\bra{1}$ and
$\Phi(\rho)=\tfrac{1}{2}\ket{0}\!\!\bra{0}+\tfrac{1}{2}\ket{2}\!\!\bra{2}$ (if $\rho$ is outside $\mathcal{D}(\mathbb{C}^3)$ scale the outputs by $\tr\rho$). Now let $\tau(\rho)=(1-\epsilon)\Pi\rho\Pi+\epsilon \rho$ with $\Pi=\ket{1}\!\!\bra{1}+\ket{2}\!\!\bra{2}$ and $\epsilon \in(0, 1).$ 
Observe that $\widehat{d}_\diamondsuit (\Psi, \Phi)=1$ but $\widehat{d}_\diamondsuit (\tau \circ \Psi, \tau \circ \Phi)=\frac{2}{1+\epsilon}>2-2\epsilon$, i.e. this $\tau$ increases the distance! See Fig. \ref{fig_non-additive}.

\begin{figure}
    \centering
    \begin{tikzpicture}
\pgfplotsset{compat=1.17}
\begin{axis}
    [
    height=3.5cm,
    width=6.5cm,
    scale only axis,
    axis lines=center,  
    xtick={0.5,1},
    xticklabels={1/2,1},
    ytick={0.5,1},
    yticklabels={1/2,1},
    ztick={0.5},
    zticklabel style={right},
    zticklabels={1/2},
    extra x ticks={0},
    xmin=-0.1,
    xmax=1.2,
    ymin=-0.05,
    ymax=1.2,
    zmin=-0.1,
    zmax=0.65,
    view={120}{20},
    after end axis/.code={
       \draw[color=black!50!white,very thick,-{Stealth}] (axis cs:0,0,0) -- (axis cs:0.5,0,0.5);
       \draw[color=black!50!white,very thick,-{Stealth}] (axis cs:0,0,0) -- (axis cs:0,0.5,0.5);
       \draw[color=blue5,very thick,-{Stealth}] (axis cs:0,0,0) -- (axis cs:1,0,0);
       \draw[color=blue5,very thick,-{Stealth}] (axis cs:0,0,0) -- (axis cs:0,1,0);
       \draw[color=blue1,style=dashed,ultra thick,-{Stealth}] (axis cs:0,0,0) -- (axis cs:0.5,0,0);
       \draw[color=blue1,style=dashed,ultra thick,-{Stealth}] (axis cs:0,0,0) -- (axis cs:0,0.5,0);
       \node[] at (axis cs:1.3,0,0.02) {$x$};
       \node[] at (axis cs:0,1.25,0) {$y$};
       \node[] at (axis cs:0,0.05,0.7) {$z$};
       }
    ]
\addplot3[
    color=black,
    style=dashed,
] coordinates
{(0.5,0,0.5) (0,0,0.5) (0,0.5,0.5)};
\addplot3[
    color=black,
    style=dashed,
] coordinates
{(0.5,0,0.5) (0.5,0,0)};
\addplot3[
    color=black,
    style=dashed,
] coordinates
{(0,0.5,0.5) (0,0.5,0)};
\end{axis}
\end{tikzpicture}
    \caption{As opposed to $d_\diamondsuit$, the renormalisation-containing $\widehat{d}_\diamondsuit$ is not contractive. We can visualize the counterexample in the text by probability vectors $p\in\mathbb{R}^3$, because all the density matrices of interest are of the form $\operatorname{diag}(p)$. The diamond distance is the $L_1$ distance of the corresponding vector pair. It \emph{grows} after the grey pair is mapped by $\tau$ ($\epsilon\to 0$) to the dashed light blue pair \emph{and} renormalized to the dark blue.}
    \label{fig_non-additive}
\end{figure}
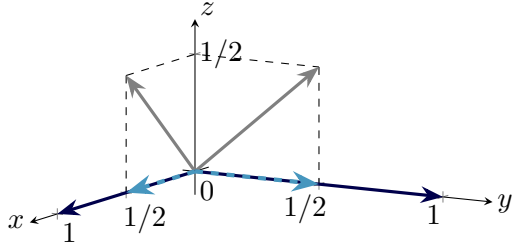

\begin{proof}[Proof of \Cref{lem:chaining_post}]
For shorthand, denote $\bar{\Psi}:=\Psi\otimes \mathcal{I}_{\K'}\otimes \mathcal{I}_{\K''}$ and similarly for $\bar{\Phi}$, denote $\bar{\Psi}':=\Psi'\otimes \mathcal{I}_{\K''}$ and similarly for $\bar{\Phi}'$.
Then by definition \eqref{eq:def_diam_post} the left-hand side of inequality (\ref{eq:chaining}) is equal to
\[
\sup_{\substack{\K''\\\rho \in \mathcal{D}
(\h\otimes\K'\otimes \K'')}}\underbrace{\left|\!\left|\frac{\bar{\Psi}'\left(\bar{\Psi}(\rho)\right)}
{\tr \left[\bar{\Psi}'\left(\bar{\Psi}(\rho)\right)\right]}
-\bar{\Phi}'\left(\frac{\bar{\Phi}(\rho)}{\tr\left[\bar{\Phi}(\rho)\right]}\right)\right|\!\right|_{tr}}_{h(\rho)}\text{,}
\]
because $\bar{\Phi}'$ is trace-preserving and linear. Let $\rho'= \bar{\Psi}(\rho)/\tr\left[\bar{\Psi}(\rho)\right]$. From the linearity of $\bar{\Psi}'$, the argument of the supremum is equal to
\begin{eqnarray*}
	h(\rho)\!&=&\left|\!\left|\frac{\bar{\Psi}'\left(\rho'\right)}
	{\tr\left[\bar{\Psi}'\left(\rho'\right)\right]}-
	\bar{\Phi}'\!\left(\frac{\bar{\Phi}(\rho)}{\tr\left[\bar{\Phi}(\rho)\right]}\right)\right|\!\right|_{tr}\\
	&\leq&\left|\!\left|\frac{\bar{\Psi}'\left(\rho'\right)}
	{\tr\left[\bar{\Psi}'\left(\rho'\right)\right]}-\bar{\Phi}'(\rho')\right|\!\right|_{tr}\!\!+\left|\!\left|\bar{\Phi}'\!\left(\rho'-\frac{\bar{\Phi}(\rho)}{\tr\left[\bar{\Phi}(\rho)\right]}\right)\right|\!\right|_{tr}\\
	&\leq&\left|\!\left|\frac{\bar{\Psi}'\left(\rho'\right)}
	{\tr\left[\bar{\Psi}'\left(\rho'\right)\right]}\!-\bar{\Phi}'(\rho')\right|\!\right|_{tr}\!\!+\left|\!\left|\frac{\bar{\Psi}(\rho)}
	{\tr\left[\bar{\Psi}(\rho)\right]}\!-\!\frac{\bar{\Phi}(\rho)}{\tr\left[\bar{\Phi}(\rho)\right]}\right|\!\right|_{tr}
\end{eqnarray*}

\noindent because $\bar{\Phi}'$ contracts the trace norm (\Cref{lem:contract}). We get (\ref{eq:chaining}) by taking the supremum.
\end{proof}

\subsection{Inequalities via conversion}

To prove inequalities in the postselected setting, sometimes it is possible to go back to the standard setting and use the results that hold there. Here we apply this method to obtain the following:

\begin{posttheorem}[$\boldsymbol{\widehat{d}_{tr}}$ small $\boldsymbol{\xxRightarrow{U} \widehat{d}_\diamondsuit}$ small]\label{lem:diam_post}
    Let $U:\h\to \h'$ be an isometry and $\Psi: L(\h)\to L(\h')$ a postselection superoperator. Assume $\widehat{d}_{tr}(\Psi, U\cdot U^\dagger) \leq \epsilon$. Then ${\widehat{d}_\diamondsuit(\Psi, U\cdot U^\dagger)\leq 24\sqrt{\epsilon}+18\epsilon}$.
\end{posttheorem}

\begin{posttheorem}[$\boldsymbol{\widehat{d}_{tr}}$ small $\boldsymbol{\xxRightarrow{U} \left|\!\left|\cdot\right|\!\right|_{op}}$ small]\label{lem:op_post}
Let ${U:\h\to \h'}$ be an isometry, $\Psi: L(\h)\to L(\h')$ a postselection superoperator. Let $A:\h\to\h'\otimes\K'$ be a Stinespring-dilation operator of $\Psi$. 
If $\widehat{d}_{tr}(\Psi, U\cdot U^\dagger) \leq \epsilon$, then there exists $\ket{g}\in\K'$, $(1-9\epsilon)\left|\!\left|A\right|\!\right|^2_{op}\leq \left|\!\left|\ket{g}\right|\!\right|^2\leq \left|\!\left|A\right|\!\right|^2_{op}$, such that $\left|\!\left|A-U\otimes \ket{g}\right|\!\right|_{op}\leq 6\left|\!\left|A\right|\!\right|_{op}\sqrt {\epsilon}$.
\end{posttheorem}

\noindent The conversion between the settings is done by the following \namecref{lem:conversion}, which we prove in the next section.

\begin{mylemma}[Conversion Lemma]\label{lem:conversion}
	Let $\Psi, \Phi: L(\h)\to L(\h')$ be postselection superoperators, with $\Phi$ trace-preserving. Then $\left|\tr\left[\Psi(\rho)\right]-\left|\!\left|\Psi\right|\!\right|_\diamondsuit\right|\leq \alpha \left|\!\left|\Psi\right|\!\right|_\diamondsuit\,\widehat{d}_{tr}\left(\Psi, \Phi\right)$ for all $\rho\in\mathcal{D}(\h)$ and
	\begin{equation}
	    \frac{1}{2}\widehat{d}_{tr}\!\left(\Psi, \Phi\right)\leq d_{tr}^\mathcal{D}\!\left(\frac{\Psi}{\left|\!\left|\Psi\right|\!\right|}_{\!\diamondsuit},\Phi\right)\leq \left(\alpha + 1\right)\widehat{d}_{tr}\!\left(\Psi, \Phi\right)\text{,}\label{eq:mainlemma}
	\end{equation}
	with $\alpha=40/\max_{\rho_0, \rho_1\in\mathcal{D}(\h)}\left|\!\left|\Phi(\rho_0)-\Phi(\rho_1)\right|\!\right|_{tr}$ in general but with $\alpha=8$ if $\Phi$ is also a unitary superoperator.
\end{mylemma}

\noindent The inequalities of \Cref{lem:conversion} hold also for diamond distances. This is immediate from the fact that the identity superoperator $\mathcal{I}_\K$ is unitary, that $\left|\!\left|\Psi\otimes\mathcal{I}_\K\right|\!\right|_\diamondsuit=\left|\!\left|\Psi\right|\!\right|_\diamondsuit$, and that taking supremum ($\sup_\K$) preserves inequalities. Thus \Cref{lem:conversion} allows us to move from postselection distances to the corresponding standard distances of linear trace non-increasing CP maps (since $\left|\!\left|\Psi\right|\!\right|_\diamondsuit=\max_{\rho\in\mathcal{D}(\h)}\tr\left[\Psi(\rho)\right]$ the superoperator $\Psi/\left|\!\left|\Psi\right|\!\right|_\diamondsuit$ is indeed trace non-increasing). We are ready to prove the \namecref{lem:diam_post}s.

\begin{proof}[Proof of \Cref{lem:diam_post}] 
Using the shorthand $D:=d_{tr}^\mathcal{D}\left(\frac{\Psi}{\left|\!\left|\Psi\right|\!\right|}_{\!\scaleto{\diamondsuit}{5pt}},U\cdot U^\dagger\right)$ we claim that
\begin{equation*}
    \frac{1}{2}\widehat{d}_\diamondsuit(\Psi, U \cdot U^\dagger)\leq
    d_\diamondsuit\left(\frac{\Psi}{\left|\!\left|\Psi\right|\!\right|}_{\!\diamondsuit},U\cdot U^\dagger\right)\leq 4\sqrt{D}+D
\end{equation*}
The first inequality is from \Cref{lem:conversion} applied to diamond distances, the second from \Cref{lem:diam}. Substituting $D\leq 9\epsilon$, also from \Cref{lem:conversion}, completes the proof.
\end{proof}

\begin{proof}[Proof of \Cref{lem:op_post}] Take $D:=d_{tr}^\mathcal{D}\left(\Psi/\left|\!\left|\Psi\right|\!\right|_\diamondsuit,U\cdot U^\dagger\right)$. Since $\Psi(\cdot)=\tr_{\,\K'}\left[A\cdot A^\dagger\right]$ and $\left|\!\left|A\right|\!\right|^2_{op} = \left|\!\left|\Psi\right|\!\right|_\diamondsuit$ (\Cref{lem:norms}), $A/\left|\!\left|A\right|\!\right|_{op}$ is a Stinespring-dilation operator of $\Psi/\left|\!\left|\Psi\right|\!\right|_\diamondsuit$. Then by \Cref{lem:op} there exists $\ket{v}\in\K'$ such that
\begin{eqnarray*}
    &1-9\epsilon \leq 1-D\leq \left|\!\left|\ket{v}\right|\!\right|^2\leq 1\text{,}&\\
    &\left|\!\left|\frac{A}{\left|\!\left|A\right|\!\right|_{op}}-U\otimes\ket{v}\right|\!\right|
    \leq 2 \sqrt{D}\leq 6\sqrt{\epsilon} \text{,}&
\end{eqnarray*}

\noindent where we used $D\leq 9 \epsilon$ from \Cref{lem:conversion}. Taking $\ket{g}:=\left|\!\left|A\right|\!\right|_{op}\ket{v}$ completes the proof.
\end{proof}

\subsection{Proof of Conversion lemma}

In this section we prove our main result, \Cref{lem:conversion}. When $\Phi$ is not promised to be unitary, the inverse dependence of $\alpha$ on the maximum trace distance of $\Phi$'s outputs is necessary. Consider the postselection superoperators $\Phi(\rho)=\ket{\boldsymbol{0}}\!\!\bra{\boldsymbol{0}}\tr\rho$ and ${\Psi(\rho)=\frac{1}{2}\ket{\boldsymbol{0}}\!\!\bra{\boldsymbol{0}}\tr\rho +\frac{1}{2}\ket{\boldsymbol{0}}\!\!\bra{\boldsymbol{0}}\rho\ket{\boldsymbol{0}}\!\!\bra{\boldsymbol{0}}}$ and note that $\Phi$ is indeed trace-preserving and that $\left|\!\left|\Psi\right|\!\right|_\diamondsuit = 1$. We have \stackMath\stackon[-8.7pt]{d_{tr}}{\widehat{\phantom{d}}\,\,}$\left(\Psi, \Phi\right)=0$, but $d_{tr}^\mathcal{D}\left(\Psi,\Phi\right)=\frac{1}{2}$, so for the inequality (\ref{eq:mainlemma}) to hold we need $\alpha\to\infty$.

\begin{proof}[Proof of \Cref{lem:conversion} (the left inequality in (\ref{eq:mainlemma}))]
The inequality actually holds also when $\left|\!\left|\Psi\right|\!\right|_\diamondsuit$ is replaced by any $k\neq 0$. To prove this, observe that for any $\rho\in\mathcal{D}(\h)$, $\Psi(\rho)$ is positive semidefinite and we have ${\left|\!\left|\Psi(\rho)\right|\!\right|_{tr}=\tr\left[\Psi(\rho)\right]}$. Therefore for all $\rho\in\mathcal{D}(\h)$
\begin{eqnarray*}
    \left|\!\left|\frac{\Psi(\rho)}{\tr\left[\Psi(\rho)\right]}-\frac{\Psi(\rho)}{k}\right|\!\right|_{tr} 
	&=& \left|1-\frac{\tr\left[\Psi(\rho)\right]}{k}\right| \\
	&=& \left|\tr\left[\Phi(\rho)-\frac{\Psi(\rho)}{k}\right]\right|\\
	&\leq& \left|\!\left|\Phi(\rho)-\frac{\Psi(\rho)}{k}\right|\!\right|_{tr}\text{,}
\end{eqnarray*}

\noindent where we used the fact that $\Phi$ is trace-preserving. From the triangle inequality and the above we get
\begin{eqnarray*}
    \left|\!\left|\frac{\Psi(\rho)}{\tr\left[\Psi(\rho)\right]}-\Phi(\rho)\right|\!\right|_{tr} &\leq& \left|\!\left|\frac{\Psi(\rho)}
	{\tr\left[\Psi(\rho)\right]}
	-\frac{\Psi(\rho)}{k}\right|\!\right|_{tr}	+\left|\!\left|\frac{\Psi(\rho)}{k}
	-\Phi(\rho)\right|\!\right|_{tr}\\
	&\leq& 2\left|\!\left|\frac{\Psi(\rho)}{k}
	-\Phi(\rho)\right|\!\right|_{tr}\text{.}
\end{eqnarray*}

\noindent Taking $\sup_\rho$ gives $\widehat{d}_{tr}(\Psi, \Phi)\leq 2 \, d_{tr}^\mathcal{D}(\frac{\Psi}{k},\Phi)$ as required.
\end{proof}

\begin{proof}[Proof of \Cref{lem:conversion} (the rest)] Use the shorthand $\widehat{D}:=\widehat{d}_{tr}(\Psi, \Phi)$. We will first prove the bound $|\tr\left[\Psi(\rho)\right]-\left|\!\left|\Psi\right|\!\right|_\diamondsuit |\leq \alpha \left|\!\left|\Psi\right|\!\right|_\diamondsuit\, \widehat{D}$. For this purpose we will use the fact that $\rho \mapsto\Phi(\rho)/\tr\left[\Phi(\rho)\right]=\Phi(\rho)$ is linear in $\rho$. For any $\rho_0, \rho_1\in\mathcal{D}(\h)$ denote $\rho_p:= (1-p)\rho_0 + p\,\rho_1=\rho_0+p\,\Delta$ where $p\in [0, 1]$ and $\Delta:=\rho_1-\rho_0$. Since $\left|\!\left|\Psi\right|\!\right|_\diamondsuit = \left|\!\left|\Psi\right|\!\right|^\mathcal{D}_{tr}=\max_{\rho\in\mathcal{D}(\h)}\tr\left[\Psi(\rho)\right]$ (see \Cref{lem:norms}), we actually want to upper bound $|\tr\left[\Psi(\Delta)\right]|$. We have
\begin{eqnarray*}
	\left|\!\left|\frac{\Psi(\rho_p)}{\tr\left[\Psi(\rho_p)\right]}
	-\Phi(\rho_p)\right|\!\right|_{tr}  \!&\leq& \widehat{D}\\
	\frac{1}{\tr\left[\Psi(\rho_p)\right]}\Big|\!\Big|\Psi(\rho_0)+p\,\Psi(\Delta)
	-\left(\tr\left[\Psi(\rho_0)\right]
	+p\tr\left[\Psi(\Delta)\right]\right)\!
	\left(\Phi(\rho_0)
	+ p\,\Phi(\Delta)\right)\!\Big|\!\Big|_{tr}\!&\leq& \widehat{D}\\
	\Big|\!\Big|A_0 + Bp - Cp^2\Big|\!\Big|_{tr}\!&\leq& \left|\!\left|\Psi\right|\!\right|_\diamondsuit \widehat{D}\text{,}
\end{eqnarray*}
\noindent where we rewrote the left-hand side as a trace norm of a polynomial in $p\in[0,1]$ with the coefficients 
	$A_0:=\Psi(\rho_0)-\tr\left[\Psi(\rho_0)\right]
	\Phi(\rho_0)\label{eq:pto0}$, 
	$B:=\Psi(\Delta)-\tr\left[\Psi(\rho_0)\right]\Phi(\Delta)-\tr\left[\Psi(\Delta)\right]\Phi(\rho_0)$ 
	and $C\!:=\!\tr\left[\Psi(\Delta)\right]\Phi(\Delta)$. 
To show that $\left|\tr\left[\Psi(\Delta)\right]\right|$ is small we will bound $\left|\!\left|C\right|\!\right|_{tr}$. Define $A_1$ analogously to $A_0$ and note that $\left|\!\left|A_i\right|\!\right|_{tr}\leq \tr\left[\psi(\rho_i)\right]\widehat{D}\leq \left|\!\left|\Psi\right|\!\right|_\diamondsuit \,\widehat{D}$. Observe also that $B=A_1-A_0+C$. By the triangle inequality
\begin{eqnarray*}
	\left|\!\left|Cp-Cp^2\right|\!\right|_{tr}&\leq& \left|\!\left|A_0+(A_1-A_0+C)p-Cp^2\right|\!\right|_{tr}+\left|\!\left|A_0(1-p)\right|\!\right|_{tr}+\left|\!\left|A_1\,p\right|\!\right|_{tr}\\
	(p-p^2)\left|\!\left|C\right|\!\right|_{tr}&\leq& \left|\!\left|\Psi\right|\!\right|_\diamondsuit \widehat{D} + (1-p)\left|\!\left|\Psi\right|\!\right|_\diamondsuit \widehat{D} + p \left|\!\left|\Psi\right|\!\right|_\diamondsuit \widehat{D}\\
	\left|\!\left|C\right|\!\right|_{tr}&\leq& \tfrac{2}{(1-p)p} \left|\!\left|\Psi\right|\!\right|_\diamondsuit \widehat{D}\text{.}
\end{eqnarray*}

\noindent Taking $p=\frac{1}{2}$ so that the bound is the tightest and substituting $\Delta=\rho_1-\rho_0$ into the definition of $C$ we get that for all pairs $\rho_0,\rho_1\in \mathcal{D}(\h)$
\begin{eqnarray}
	\left|\tr\left[\Psi(\rho_1)\right]-\tr\left[\Psi(\rho_0)\right]\right|
	&\leq \frac{8 \left|\!\left|\Psi\right|\!\right|_\diamondsuit\,\widehat{D}}
	{\left|\!\left|\Phi(\rho_1)-\Phi(\rho_0)\right|\!\right|_{tr}}\text{.} \label{eq:pairbound}
\end{eqnarray}

If $\Phi$ is not unitary, denote by $\rho^*_0, \rho_1^*$ the two states that maximize the denominator, i.e. $\Phi(\rho^*_0)$ and $\Phi(\rho^*_1)$ are the furthest away in trace norm. Call their distance $s$. 
Now for any state $\rho\in\mathcal{D}(\h)$, $\Phi(\rho)$ is more than $\frac{s}{2}$-far either from $\Phi(\rho_0^*)$ or from $\Phi(\rho_1^*)$, or both. Therefore, for all pairs $\rho, \rho'\in\mathcal{D}(\h)$ there exist $i,j\in\{0, 1\}$ such that
\begin{eqnarray}
	\big|\tr\left[\Psi(\rho)\right]-\tr\left[\Psi(\rho')\right]\big| 
	&\leq& \left|\tr\left[\Psi(\rho)\right] -\tr\left[\Psi(\rho_{i}^*)\right]\right|	+ \left|\tr\left[\Psi(\rho_{i}^*)\right] -\tr\left[\Psi(\rho_{j}^*)\right]\right|+\left|\tr\left[\Psi(\rho_{j}^*)\right] -\tr\left[\Psi(\rho')\right]\right|	\nonumber\\
	&\leq& \tfrac{8\left|\!\left|\Psi \right|\!\right|_{\diamondsuit}\widehat{D}}{s/2} + \tfrac{8\left|\!\left|\Psi \right|\!\right|_{\diamondsuit}\widehat{D}}{s}  + \tfrac{8\left|\!\left|\Psi \right|\!\right|_{\diamondsuit}\widehat{D}}{s/2}\nonumber\\ &=&\frac{40}{s} \left|\!\left|\Psi \right|\!\right|_{\diamondsuit}\widehat{D}	\text{.} \label{eq:conversion_nonunitary}
\end{eqnarray}

If $\Phi$ is a unitary superoperator, we can obtain a tighter bound, because then for any orthogonal pure states $\rho, \rho^\perp\in\mathcal{D}(\h)$ we have $\left|\!\left|\Phi(\rho)-\Phi(\rho^\perp)\right|\!\right|_{tr}=\left|\!\left|\rho-\rho^\perp\right|\!\right|_{tr}=2$ and Eq. (\ref{eq:pairbound}) gives $\left|\tr\left[\Psi(\rho)\right]-\tr\left[\Psi(\rho^\perp)\right]\right|\leq 4 \left|\!\left|\Psi \right|\!\right|_{\diamondsuit}\widehat{D}$.
Now note that for any pure state $\rho\in\mathcal{D}(\h)$ we can write the completely mixed state as $\frac{1}{d}\id=\frac{1}{d}\rho +\frac{1}{d}\sum_{i=1}^{d-1}\rho_i^\perp$, with $d=\dim(\h)$ and with each $\rho^\perp_{i}$ pure and orthogonal to $\rho$. We get that
\begin{samepage}
\begin{eqnarray}
	 \left|\tr\left[\Psi(\rho)\right]
	 -\tr\left[\Psi\left(\tfrac{1}{d}\id\right)\right]\right|
	 &\leq& \tfrac{1}{d}\sum_{i=1}^{d-1}\left|\tr\left[\Psi(\rho)\right]
	 -\tr\big[\Psi(\rho_i^\perp)\big]\right|\nonumber\\
	 &\leq& 4\left|\!\left|\Psi \right|\!\right|_{\diamondsuit}\widehat{D}\label{eq:conversion_unitary}
\end{eqnarray}
holds for any pure $\rho$. By the convexity of $\mathcal{D}(\h)$, it in fact holds for any $\rho\in\mathcal{D}(\h)$.\\
\end{samepage}

To continue both the unitary and the non-unitary case, choose $\rho'$ that maximizes $\tr\left[\Psi(\rho')\right]$. Then for any  $\rho\in\mathcal{D}(\h)$
\begin{equation*}
     \left|\tr\left[\Psi(\rho)\right]-\left|\!\left|\Psi\right|\!\right|_\diamondsuit\right|
     =\left|\tr\left[\Psi(\rho)\right]-\tr\left[\Psi(\rho')\right]\right|
     \leq \alpha \left|\!\left|\Psi \right|\!\right|_{\diamondsuit} \widehat{D} 
\end{equation*}
with $\alpha$ as in the statement of \Cref{lem:conversion} (the non-unitary case follows from (\ref{eq:conversion_nonunitary}), the unitary from (\ref{eq:conversion_unitary}) and the triangle inequality). This proves the first conclusion of \Cref{lem:conversion}. To prove the right inequality in (\ref{eq:mainlemma}) note that $\left|\!\left|\Psi(\rho)\right|\!\right|_{tr}=\tr\left[\Psi(\rho)\right]$ from $\Psi$'s complete positivity, and therefore $\left|\!\left|\frac{\Psi(\rho)}{\left|\!\left|\Psi \right|\!\right|_{\diamondsuit}}
    -\frac{\Psi(\rho)}{\tr\left[\Psi(\rho)\right]}\right|\!\right|_{tr} = 
    \left|\frac{\tr\left[\Psi(\rho)\right]}{\left|\!\left|\Psi \right|\!\right|_{\diamondsuit}}
    -1\right| \leq \alpha \widehat{D}
$. We get that for all $\rho \in \mathcal{D}(\h)$
\begin{eqnarray*}
	\left|\!\left|\frac{\Psi(\rho)}{\left|\!\left|\Psi \right|\!\right|_{\diamondsuit}}-
	\Phi(\rho)\right|\!\right|_{tr}&\leq&
	\left|\!\left|\frac{\Psi(\rho)}{\left|\!\left|\Psi \right|\!\right|_{\diamondsuit}}-\frac{\Psi(\rho)}
	{\tr\left[\Psi(\rho)\right]}\right|\!\right|_{tr}+\left|\!\left|\frac{\Psi(\rho)}{\tr\left[\Psi(\rho)\right]}
	-\Phi(\rho)\right|\!\right|_{tr}\\
	&\leq& \alpha\widehat{D} + \widehat{D}\text{.} 
\end{eqnarray*}
\noindent Taking $\sup_\rho$ completes the proof.
\end{proof}

\section*{Acknowledgments}
The author would like to thank Dorit Aharonov, Itai Leigh, and Matan Seidel for useful discussions. This research was supported by Simons Foundation (grant 385590) and Israel Science Foundation (grants 2137/19 and 1721/17).

\bibliographystyle{unsrt}
\bibliography{main} 

\providecommand{\noopsort}[1]{}\providecommand{\singleletter}[1]{#1}%
\begin{thebibliography}{10}

\bibitem{knill2001scheme}
Emanuel Knill, Raymond Laflamme, and Gerald~J Milburn.
\newblock A scheme for efficient quantum computation with linear optics.
\newblock {\em Nature}, 409(6816):46--52, 2001.

\bibitem{knill2005quantum}
Emanuel Knill.
\newblock Quantum computing with realistically noisy devices.
\newblock {\em Nature}, 434(7029):39--44, 2005.

\bibitem{reichardt2006error}
Ben~W Reichardt.
\newblock Error-detection-based quantum fault tolerance against discrete pauli
  noise.
\newblock Preprint at \url{https://arxiv.org/abs/quant-ph/0612004}, 2006.

\bibitem{aliferis2008accuracy}
Panos Aliferis, Daniel Gottesman, and John Preskill.
\newblock Accuracy threshold for postselected quantum computation.
\newblock {\em Quantum Information \& Computation}, 8(3):181--244, 2008.

\bibitem{aaronson2005quantum}
Scott Aaronson.
\newblock Quantum computing, postselection, and probabilistic polynomial-time.
\newblock {\em Proceedings of the Royal Society A: Mathematical, Physical and
  Engineering Sciences}, 461(2063):3473--3482, 2005.

\bibitem{bremner2011classical}
Michael~J Bremner, Richard Jozsa, and Dan~J Shepherd.
\newblock Classical simulation of commuting quantum computations implies
  collapse of the polynomial hierarchy.
\newblock {\em Proceedings of the Royal Society A: Mathematical, Physical and
  Engineering Sciences}, 467(2126):459--472, 2011.

\bibitem{bremner2017achieving}
Michael~J Bremner, Ashley Montanaro, and Dan~J Shepherd.
\newblock Achieving quantum supremacy with sparse and noisy commuting quantum
  computations.
\newblock {\em Quantum}, 1:8, 2017.

\bibitem{aaronson2011computational}
Scott Aaronson and Alex Arkhipov.
\newblock The computational complexity of linear optics.
\newblock In {\em Proceedings of the forty-third annual ACM symposium on Theory
  of computing}, pages 333--342, 2011.

\bibitem{morimae2014hardness}
Tomoyuki Morimae, Keisuke Fujii, and Joseph~F Fitzsimons.
\newblock Hardness of classically simulating the one-clean-qubit model.
\newblock {\em Physical review letters}, 112(13):130502, 2014.

\bibitem{bouland2019complexity}
Adam Bouland, Bill Fefferman, Chinmay Nirkhe, and Umesh Vazirani.
\newblock On the complexity and verification of quantum random circuit
  sampling.
\newblock {\em Nature Physics}, 15(2):159--163, 2019.

\bibitem{beverland2020lower}
Michael~E Beverland, Earl Campbell, Mark Howard, and Vadym Kliuchnikov.
\newblock Lower bounds on the non-clifford resources for quantum computations.
\newblock {\em Quantum Science and Technology}, 2020.

\bibitem{aaronson2019online}
Scott Aaronson, Xinyi Chen, Elad Hazan, Satyen Kale, and Ashwin Nayak.
\newblock Online learning of quantum states.
\newblock {\em Journal of Statistical Mechanics: Theory and Experiment},
  2019(12):124019, 2019.

\bibitem{gao2015quantum}
Jingliang Gao.
\newblock Quantum union bounds for sequential projective measurements.
\newblock {\em Physical Review A}, 92(5):052331, 2015.

\bibitem{aaronson2006qma}
Scott Aaronson.
\newblock $\text{QMA}$/qpoly $\subseteq$ $\text{PSPACE}$/poly:
  de-$\text{M}$erlinizing quantum protocols.
\newblock In {\em 21st Annual IEEE Conference on Computational Complexity
  (CCC'06)}, pages 13--pp. IEEE, 2006.

\bibitem{kitaev1997quantum}
A~Yu Kitaev.
\newblock Quantum computations: algorithms and error correction.
\newblock {\em Russian Mathematical Surveys}, 52(6):1191, 1997.

\bibitem{mixed_states_aharonov98}
Dorit Aharonov, Alexei Kitaev, and Noam Nisan.
\newblock Quantum circuits with mixed states.
\newblock In {\em Proceedings of the thirtieth annual ACM symposium on Theory
  of computing}, pages 20--30, 1998.

\bibitem{watrous2018theory}
John Watrous.
\newblock {\em The theory of quantum information}.
\newblock Cambridge University Press, 2018.

\bibitem{watrous2005notes}
John Watrous.
\newblock Notes on super-operator norms induced by schatten norms.
\newblock {\em Quantum Information \& Computation}, 5(1):58--68, 2005.

\bibitem{stinespring1955positive}
W~Forrest Stinespring.
\newblock Positive functions on $\text{C}^*$-algebras.
\newblock {\em Proceedings of the American Mathematical Society},
  6(2):211--216, 1955.

\end{thebibliography}
\end{document}